\documentclass{article}

\usepackage{amsmath,amssymb}
\usepackage{amsthm}

\usepackage[usenames]{color}
\usepackage{rotating}
\usepackage{setspace}
\usepackage{algorithm2e}

\newtheorem{definition}{Definition}
\newtheorem{proposition}[definition]{Proposition}
\newtheorem{example}[definition]{Example}
\newtheorem{theorem}[definition]{Theorem}

\begin{document}

\title{Reidentification and k-anonymity: \\a model for disclosure risk in graphs}
\author{
  Klara Stokes$^{1}$, Vicen\c c Torra$^{2}$\\
\small $^1$ Universitat Rovira i Virgili\\
\small  Dept. of Computer Engineering and Maths, \\
\small  UNESCO Chair in Data Privacy\\
\small  Av. Pa\"{\i}sos Catalans 26, 43007 Tarragona, Catalonia, Spain\\
\small  Email: klara.stokes@urv.cat\\ 
\small $^2$ IIIA, Institut d'Investigaci\'o en Intel$\cdot$lig\`encia Artificial\\ 
\small  CSIC, Consejo Superior de Investigaciones Cient\'{\i}ficas\\ 
\small  Campus UAB, 08193 Bellaterra, Catalonia, Spain\\
\small  Email: vtorra@iiia.csic.es\\ 
}

\maketitle

\begin{abstract}
In this article we provide a formal framework for reidentification in general. 
We define $n$-confusion as a concept for modelling the anonymity of a database table and we prove that $n$-confusion is a generalization of $k$-anonymity. 
After a short survey on the different available definitions of $k$-anonymity for graphs we provide a new definition for $k$-anonymous graph, which we consider to be the correct definition. 
We provide a description of the $k$-anonymous graphs, both for the regular and the non-regular case. 
We also introduce the more flexible concept of $(k,l)$-anonymous graph. 
Our definition of $(k,l)$-anonymous graph is meant to replace a previous definition of $(k,l)$-anonymous graph, 
which we here prove to have severe weaknesses. 
Finally we provide a set of algorithms for $k$-anonymization of graphs. 

\end{abstract}

\section{Introduction}
In data privacy, the assessment of risk is one of the elements of major importance. At present, several approaches have been studied in the literature. The major approaches are k-anonymity~\cite{{ref:Samarati.2001},{ref:Samarati.Sweeney.1998},{ref:Sweeney:2002:algorithm},{ref:Sweeney.2002:IJUFKS:def}}, reidentification~\cite{{ref:Domingo.Torra.2001:DRExperiments:Wash}, {ref:Winkler.2004:PSD}} and differential privacy~\cite{ref:Dwork.2006}. 

In this paper, we focus on two of them: reidentification and $k$-anonymity. The former evaluates the disclosure risk of a protected data set measuring the chances that an intruder can link his information and the one in the protected data set. In contrast, $k$-anonymity tries to avoid any reidentification, producing a protected data set where each record is cloaked into a set of other $k-1$ records. 

In this work we formalize the reidentification process and we use this formalization to discuss the concept of $k$-anonymity and then propose the concept of $n$-confusion. We also prove that $n$-confusion is a generalization of $k$-anonymity. 

Then, we discuss the application of the concept of $k$-anonymity to graphs. At present, due to the interest of online social networks, several authors have studied data protection for graphs. It is relevant here to consider the works in~\cite{{ref:Feder.Nabar.Terzi.2008},{ref:Hay.Miklau.Jensen.2008},{ref:Liu.Terzi.2008},{ref:Zhou.Pei.2008}}, in which alternative definitions of $k$-anonymity for graphs have been presented. In this work we discuss these different definitions and we show that the definition in~\cite{ref:Feder.Nabar.Terzi.2008} have severe weaknesses. Then, we first provide an alternative definition for $k$-anonymity that provides enough security at the cost of being quite restrictive, and later the definition of $(k,l)$-anonymity, a relaxation of the former one. 

The paper discusses several properties of the definitions. In particular, we study the characterization of the $k$-anonymous graphs. We also provide algorithms for transforming a graph into a $k$-anonymized one, for calculating the degree of $(k,l)$-anonymity of a graph given $k$, and to increase the $l$ of the $(k,l)$-anonymity of a graph. 

The structure of the paper is as follows. Section~\ref{sec:2} discusses disclosure risk on online social networks focusing on reidentification and $k$-anonymity. Section~\ref{sec:3} reviews previous approaches of $k$-anonymity for graphs, and presents an attack for the approach introduced in~\cite{ref:Feder.Nabar.Terzi.2008}. Section~\ref{sec:4} introduces a new definition for $k$-anonymity and studies some properties of this definition. Section~\ref{sec:5} presents a relaxation $(k,l)$-anonymity. Section~\ref{sec:6} includes the algorithms that we have developed related to $(k,l)$-anonymity. The paper finishes with the conclusions. 

\section{Disclosure risk evaluation}
\label{sec:2}
In this section we first present a formal framework for evaluating disclosure risk in data privacy in general (see also~\cite{PAIS2012}). 
Then we will focus on disclosure risk for online social networks.  
\subsection{Reidentification in privacy protected databases}
A database is a collection of records of data.
In this article we will suppose that all records correspond to distinct individuals or objects.  
Every record has a unique identifier and is divided into attributes. 
The attributes can be very specific, as the attributes ``height'' or ``gender'', or more general, as the attributes ``text'' or ``sequence of binary numbers''. 

Suppose that the database can be represented as a single table. 
Let the records be the rows of the table and let the attributes be the columns. 
The intersection of a row and an attribute is a cell in the table, and we call the data in the cells the entries of the database.  

Let $T$ be a table with $n$ records and $m$ attributes. 
We define the partition set $\mathcal{P}(T)$ of $T$ to be the set of subsets of the underlying set of entries of $T$, $\bigcup_{i,j=1}^{n,m}T[i,j].$

\begin{definition}\label{1:def:methodanonymization}
A method for anonymization of databases is any transformation or operator
$$\begin{array}{rllll}\rho:&D&\rightarrow &D\\\\
&X&\mapsto&Y,
\end{array}$$
where $D$ is a space of databases. 
\end{definition}
Then $\rho$, given a database $X$, returns a database $Y$. 
Since $Y$ is a database, all entries in $Y$ will correspond to a unique individual or object, which we will suppose to be the same individuals as the ones behind the records in $X$. 
Usually it is assumed that there is, in some sense, less sensible information about the individuals behind the records in $X$ in the transformed database $Y$ than there was in the 
original database $X$. 

\begin{definition}\label{1:def:reidentification}
Let $\rho$ be a method for anonymization of databases, $X$ a table with $n$ records indexed by $I$ in the space of tables $D$ and $Y=\rho(X)$ the anonymization of $X$ using $\rho$. 
Then a re-identification method is a function that 
given a collection of entries $y$ in $\mathcal{P}(Y)$ and some additional information from a space of auxiliary informations $A$, 
returns the probability that $y$ correspond to entries from the record with index $i\in I$,
$$\begin{array}{rccc}r:&\mathcal{P}(Y) \times A &\rightarrow &[0,1]^{n}\\\\
&(y,a) &\mapsto &\left(P(y\mbox{ correspond to entries from } X[i]): i\in I\right).
\end{array}$$
Consider the objective probability distribution corresponding to the re-identification problem. Then, we require from a re-identification method that it returns a probability distribution that is compatible with this probability, also after missing some relevant information. Compatibility can be modeled in terms of compatibility of belief functions (see \cite{{ref:Chateauneuf.1994},{ref:Smets.Kennes.1994}}). 
\end{definition}
In this definition, the probability $r(y,a)$ could have been expressed as conditioned by $a$. However, in this article we prefer to use the notation above.

Compatibility implies that the more evidence we have, the less entropy the probability has. Because of that, the method returns $(1/n, \dots, 1/n)$ if it has no evidence of the original record corresponding to the protected record $y$, and that, as the evidence increases, the probabilities of the corresponding indices are agumented. An example of this situation is when the variables of a data set are protected independently by means of k-anonymity. Then, re-identification can be applied to protected data using only some of the attributes. The distributions computed by these methods should be compatible with the ones considered when all attributes are taken into account. When no attribute is considered, the algorithm should lead to the probability $(1/n, \dots, 1/n)$.
 
In this way we avoid re-identification methods with false positives, since these  would disturb the rest of the discussion. Also, probabilities are assumed to be defined so that the same value will apply to different protected records whenever these have the same value.

A common assumption is to consider that re-identification occurs when the probability function that is returned by the re-identification method takes the value 1 at one index, say at $i_0$, and the value 0 at all the other indices. That is, given the auxiliary information $a$ there is probability 1 that $y$ belongs to the record with index $i_0$ in $X$. 

We say that the entries $s\in \mathcal{P}(Y)$ are linked to a collection of indices $J\subseteq I$ if the probabilities that are returned by the re-identification method take non-zero values over the indices $J$ and are zero on the complement $I\setminus J$. In a nice and regular situation, a possible non-zero value for the re-identification method over $J$ is then $1/|J|$. 

Here, auxiliary information denotes any information used to achieve a better performance of the re-identification process. It is common that researchers use parts of the original database $X$ as such auxiliary information. One can however argue that for example knowledge of the method of anonymization is also auxiliary information. 
When the database covers only part of a population, and it is not known who is in the database, then information about individuals that are not in the original database can serve as auxiliary information. 
In general, we do not assume that the auxiliary information can be indexed by individuals or that it has any particular structure at all.

\begin{definition}
We define the confusion of a method of re-identification $r$, 
with respect to the anonymized database $Y=\rho(X)$, 
the auxiliary information $A$ and the threshold $t\in [0,1]$ as 
$$\mbox{\emph{conf}}(r,Y,A,t)=\inf_{y\in Y} M_{r,t}(y,A)$$
where 
$$M_{r,t}(y,A)=\left\{\begin{array}{ll}\inf_{a\in A}|\{i\in I:r(y,a)[i]\geq t\}|    \vspace{0.2cm} \\ 
~~~~~~~~~~~~~~~~~~~~~\mbox{if } \inf_{a\in A}|\{i\in I:r(y,a)[i]\geq t\}|>0\\  \vspace{0.3cm} \\ 
|I|~~~~~~~~~~~~~~~~~~\mbox{if } \inf_{a\in A}|\{i\in I:r(y,a)[i]\geq t\}|=0
\end{array}\right.$$
\end{definition}

The space of auxiliary information contains any information that could be useful and accessible to the adversary. 
Examples of auxiliary information is information in the public domain and information on the method that was used in order to anonymize the data.  A bad determination of $A$ could imply that the confusion of a re-identification method is overestimated. 

We have assumed that we know who is in the database and who is not. 
It can be proved that, under this assumption, any additional information that is useful for a re-identification method can be deduced from the original database $X$. 
Hence, in this particular case, we may assume that $A=X$. 

\begin{definition}
\label{def:ntconfusion}
Given a space of databases $D$, a space of auxiliary information $A$ and a method of anonymization of databases $\rho$, 
we say that $\rho$ provides $(n,t)$-confusion if for all re-identification methods $r$, 
and all anonymized databases $\rho(X)\in D$ the confusion of $r$ with respect to $\rho(X)$ and $A$ is larger or equal to $n$ for the fixed threshold $0<t\leq 1/n$.  
\end{definition}


The $(n,t)$-confusion therefore measures the smallest cardinality of a set of individuals for which the re-identification methods gives probability higher than $t$, calculated for all protected registers $y\in \rho(X)$. 

$(n,t)$-Confusion is more reliable as a measure of anonymity when $n$ and $t$ are both large, simultaneously. That is, an anonymization method providing $(n',t')$-confusion is better or equal than one that provides $(n,t)$-confusion when $t'\geq t$ and $n'\geq n$. This statement is based on the following observations. 

On the one hand, if $n$ is small, then an adversary might be able to form a collusion of size $n-1$ which could break the anonymity of the $n$th register. This issue is analogous to what we get if we implement $k$-anonymity with a small $k$. 

On the other hand, it is interesting to observe what may occur if the threshold $t$ is much smaller than $1/n$. For example, say that we apply all available re-identification methods to the protected record $y$. Suppose that the best result gives us $r(y,a)[i_0]=0.9$ and that there are $n-1$ other indices $i$ for which $r(y,a)[i]=0.1/(n-1)$. Then we have $(n,t)$-confusion with $t=0.05/(n-1)$. To avoid this issue, an adequate value for $t$ could be approximately $t=1/n$. 

Next we present a proposal for a definition of $n$-confusion.
\begin{definition}
Let notations be as in Definition~\ref{def:ntconfusion}. 
We say that the anonymization method $\rho$ provides $n$-confusion if there is a $t>0$ such that $\rho$ provides $(n,t)$-confusion.  
\end{definition} 

This definition of $n$-confusion is designed so that $n$ will be the measure of the smallest number of indices in $1,\dots,n$ for which the result is a non-zero probability, when applying a re-identification method to a protected record in $\rho(X)$. In this way, anonymization methods for which re-identification returns very distinct probabilities also provide $n$-confusion, whenever at least $n$ of these probabilities are non-zero for every protected record. This mimics $k$-anonymity in the sense that in order to re-identify an individual with absolute certainty, a collusion of the $k-1$ other target individuals is necessary.

Other possible approaches are to consider that a table satisfies $n$-confusion
\begin{itemize}
\item when the highest value of the best re-identification method is taken at least $n$ times, for every protected record $y$ and all auxiliary information;
\item when the highest value of the best re-identification method is at most $1/n$. 
\end{itemize}

\subsection{An approach for disclosure risk control: k-anonymity}
The concept of $k$-anonymity~\cite{{ref:Samarati.2001},{ref:Samarati.Sweeney.1998},{ref:Sweeney:2002:algorithm},{ref:Sweeney.2002:IJUFKS:def}} 
encompasses a set of techniques for data protection that try to avoid reidentification risk. 
When protecting a database, we have $k$-anonymity when a record is cloaked into a set of other $k-1$ records. 
Thus, given the record of the intruder, reidentification always returns $k$ indistinguishable records.

We use the notation $T(A)$ to say that $T$ is a table with the set of attributes $A$. 
Let $B\subseteq A$ be a set of attributes of the table. 
We denote the projection of the table on the attributes $B$ by $T[B]$.
We suppose that every record contains information about a unique individual. 
An identifier $I$ in a database is an attribute such that it uniquely identifies the individuals behind the records. 
In particular, any entry in $T[I]$ is unique. 
A quasi-identifier $QI$ in the database is a collection of attributes $\{A_1,\dots,A_n\}$ that belongs to the public domain (i.e. are known to an adversary), such that they in combination can uniquely, or almost uniquely, identify a record \cite{Dalenius1986}. 
That is, the structure of the table allows for the possibility that an entry in $T[QI]$ is unique, or that there are only a small number of equal entries. 
In the former case the entry in $T[QI]$ uniquely identifies the individual behind the record and in the latter, the few other individuals with the same entries in $T[QI]$ may form a collusion and use secret information about themselves in order to make this identification possible.  

The former case may be formalized as follows. 
Consider the table $\tilde{T}$ obtained by permuting randomly the records of $T$. 
Let $s$ be an element in $\mathcal{P}(T)$ such that the entries of $s$ all belong to the same record in $\tilde{T}$ (and therefore also in $T$). 
Then, if there is a method of reidentification $r:\tilde{T}\times A\rightarrow [0,1]^n$ such that $r(s,a)[i]=1$ for some $a\in A$ and one index $i$, 
then $s$ belongs to $T[QI]$.

In the latter case, an $s$ such that $r(s,a)$ is large for a small subset $J$ of indices and 0 for the others (so that $s$ is linked to $J$) would also belong to $T[QI]$. 

\begin{example}
If a table contains information on students in a school class, the attributes birth data and gender could be sufficient to determine to which individual a record of the table corresponds, although it is possible that not all records will be uniquely identified in this way. 
Hence for this table, birth date and gender are an example of a quasi-identifier.
\end{example}
The following definition of $k$-anonymity appeared for the first time in~\cite{ref:Samarati.Sweeney.1998} (see also the articles by Samarati~\cite{ref:Samarati.2001} and Sweeney \cite{ref:Sweeney.2002:IJUFKS:def}).
\begin{definition}\label{1:def:nan}
A table $T$, that represents a database and has associated quasi-identifier $QI$, is $k$-anonymous if every sequence in $T[QI]$ appears with at least $k$ occurrences in $T[QI]$. 
\end{definition}

We have the following relation between $k$-anonymity and $n$-confusion. 

\begin{theorem}\label{thm:anonym-confusion}
Consider a database $X$ and a space of auxiliary information $A$. 
Suppose that our knowledge on $A$ has permitted us to determine correctly the quasi-identifier $QI$ of $X$. 
Let $\rho$ be a method of anonymization of databases that gives $k$-anonymity with respect to $QI$. 
Then there is a threshold $0<t<1$ such that any method of reidentification $r$ is of confusion at least $k$ with respect to $\rho(X)$, $r$, and $t$, so that $\rho$ provides $k$-confusion. 
\end{theorem}
\begin{proof}
Let $X$ be a table with records indexed by $I$ and $Y=\rho(X)$ a table that is $k$-anonymous with respect to $QI$, obtained by applying the anonymization method $\rho$ to $X$. 
The assumption that our knowledge on $A$ has permitted us to determine correctly the quasi-identifier $QI$, implies that we may express the auxiliary information as the restriction of the original table to the quasi-identifier, $A=X[QI]$. 
Then any reidentification method $r:Y\times A\rightarrow [0,1]^{|I|}$ will take values $r(y,a)$ in which at least $k$ entries $r(y,a)[i]$ are equal to $x$ with $0<x\leq 1/k$ and the other entries are smaller than $x$. 
Define $t$ as the minimum among all these $x$. 
Then the confusion of $r$ is at least $k$, for the threshold $0<t\leq 1/k$. 
Since $X$ was any table and $r$ any method of reidentification, $\rho$ provides $k$-confusion.  
\end{proof}
In some cases the threshold for a $k$-anonymous table will be $1/k$. 

Theorem~\ref{thm:anonym-confusion} shows that any table that satisfies $k$-anonymity also satisfy $n$-confusion with $n:=k$. 
The converse is not true, that is, a table that satisfies $n$-confusion does not necessarily satisfy $n$-anonymity. 
Next we will see an example of this. 

\begin{example}
Let $X$ be a numerical table with 30 distinct records (points) in $\mathbb{R}^3$. 
Suppose that we want to anonymize $X$ according to $k$-anonymity with $k=3$. 
A common approach for achieving $k$-anonymity is to apply a clustering algorithm to the table, see for example~\cite{DomingoTorra}.
A clustering algorithm returns a partition of the records so that the records in each class of the partition are in some sense similar. 
In this case the clustering algorithm returns a partition of the record set in 10 classes with exactly 3 records in each class, so that these records are points inside a ball of radius $r$ from the average of the three points. 
Given the records (points) $p_1,p_2,p_3\in \mathbb{R}^3$, let $A(\{p_1,p_2,p_3\})$ represent their average, and $V(\{p_1,p_2,p_3\})$ represent the normalized vector that is perpendicular to the plane defined by the points $p_1,p_2,p_3$. 
Let $c(p)$ represent the points in the cluster of the point $p$. 
Two alternatives are considered for the definition of a $3$-anonymization of $X$. 
\begin{enumerate}
\item Replace $p\in X$ by $A(c(p))$;
\item Let $p, p', p''$ be the three points in a cluster $c$. 
Replace $p$ by $A(c)$, $p'$ by $A(c)+\epsilon V(c)$ and $p''$ by $A(c)-\epsilon V(c)$, 
where $\epsilon$ is a positive real number smaller than the radius $r$ of the ball that contains the points in the cluster $c$.
\end{enumerate}
Then, the first alternative satisfies both $3$-anonymity and $(3,1/3)$-confusion. 
However, the second alternative does not satisfy $3$-anonymity, but it does satisfy $(3,1/3)$-confusion. 
\end{example}

\subsection{k-Anonymity for graphs}
A graph is a pair $(V,E)$, where $V$ is a set of vertices and $E$ is a family of 2-subsets of $V$ called edges. 
Sometimes there is associated an additional information to a vertex or to an edge. Such information is called a label. 
Graphs can be used to represent for example a social network. 
In the common approach for representing a social network as a graph, individuals are represented as vertices, information about relations between individuals is represented as edges and other information about the individuals or about the relations is represented as labels.  

The concept of $k$-anonymity was initially defined for tables. 
In order to apply $k$-anonymity to other data structures, observe that these can be represented in table form. 
For example, when applying $k$-anonymity to graphs, the adjacency matrix of the graph is a representation of the graph in table form. 
The adjacency matrix of a graph is a matrix in which both the rows and the columns are indexed by the vertices of the graph and the entries represent the number of edges between the corresponding vertices. 
A table is obtained by taking the rows and the columns of the matrix to be the records and the attributes of the table, respectively. 
Then every vertex of the graph occurs as an index of a record in the table and the attributes indicate the existence of edges to the other vertices. 
Depending on the situation, what is considered to be interesting information about the graph may vary. 
Therefore, other attributes may be included in the table, like for example the degree of the vertices.  
The information given by the adjacency matrix is however enough to deduce any other information available about the graph.
There are also other representations of graphs that contain the same information as the table just described. 
One example is the incidence matrix, in which the rows are indexed by the vertices, the columns are indexed by the edges and the entries indicate if the correponding  vertex is on the corresponding edge. 

As we saw in the previous section, the concept of $k$-anonymity is based on the existence of a quasi-identifier. 
A quasi-identifier is a collection of attributes of the table, and the quality of the anonymization of the table depends on the correct determination of the quasi-identifier. 
Usually it is the data owner who is the entity that executes the anonymization method and who therefore is the responsible for the correct determination of the quasi-identifier. 
In this process, the data owner must consider a table form representation of the graph. 
The choice of attributes for this table is of course crucial for the determination of the quasi-identifier. 
The data owner can not know in advance which information may be useful for the adversary. 
For example, the adversary could use the edge set of the vertices for the reidentification process, or he could use only the degree of the vertices. 
In the former example the adversary uses exactly the information given by the adjacency matrix, and in the latter example he uses information that can be derived from the adjacency matrix, and that can be attacked to the table as an additional attribute.  
A prudent data owner should assume that the table that represents the graph, and which serves for the determination of the quasi-identifier, contains all attributes that may be relevant for a reidentification intent of an adversary. 
We observe that this will occur exactly when the table that is defined for the adjacency matrix (or some other equivalent table) is $k$-anonymized!
This is so, because all other information about the graph can be deduced from this table. 

\subsection{Table data k-anonymity versus graph k-anonymity} 
Previous authors on this subject seem to agree on the opinion that $k$-anonymization for graphs differs from $k$-anonymization of tables. 
This opinion complicates the application of the concept of $k$-anonymity to graphs. 
We argue that graph $k$-anonymity is a special case of $k$-anonymity as defined by Sweeney.

Below we list the arguments used in  \cite{ref:Zhou.Pei.2008} to justify the difference between $k$-anonymization of table form data and graph data. 
\begin{enumerate}
\item They claim that it is much more challenging to model the background
knowledge of adversaries and attacks about social network
data than that about relational data. On relational data, they say, it is often
assumed that a set of attributes serving a quasi-identifier is
used to associate data from multiple tables, and attacks mainly
come from identifying individuals from the quasi-identifier.
However, in a social network many pieces of information
can be used to identify individuals, such as labels of vertices
and edges, neighborhood graphs, induced subgraphs, and their
combinations. It is much more complicated and much more
difficult than the relational case.

\item They also claim that is much more challenging to measure the information
loss in anonymizing social network data than that in
anonymizing relational data.

\item Finally, they claim that it is much more difficult to anonymize a
social network than data in table form, since changing labels
of vertices and edges may affect the neighborhoods of other
vertices, and removing or adding vertices and edges may
affect other vertices and edges as well as the properties of
the network.
\end{enumerate} 
Observe that only the first point is relevant for the definition of $k$-anonymity,  since it focus on the choice of quasi-identifier, at least if we consider this choice to form part of the definition of $k$-anonymity for graphs.  
The second and the third points are only important when defining an algorithm for $k$-anonymization of graphs. 
In particular the third point is of a completely practical nature. Also, it is important to realize that every kind of data has its peculiarities.
Also relational table data can hide unexpected quasi-identifiers, caused by the origin of the data. 
Trajectorial data can be tricky, when anonymizing car trajectories it is important to check that the published trajectories are feasible. For example, a car can not drive over a lake. 
For graphs we have a similar situation. 
When anonymizing a graph it is important to check that we do not produce edges which contain only one vertex. 
Concluding, we see that the anonymization algorithm must take into account the underlying structure of the data type that is being anonymized. 

Regarding the first point, we have seen in the previous discussion that all available information about a graph can be deduced from its adjacency matrix. 
This implies that if the matrix is $k$-anonymous, so is the graph, with respect to any graph property that can come in mind. 

\subsection{n-Confusion for graphs}
We will here make some remarks on $n$-confusion for graphs, although we spare a more detailed discussion on this subject for future work. 

The concept of $n$-confusion generalizes $k$-anonymity, and permits to define methods of anonymization that do not provide  $k$-anonymity, but that do provide the same level of anonymity as does $k$-anonymity. The main interest is then to minimize information loss. 
Just as for any table data, $n$-confusion can be used for privacy protection of data from social networks. 
We can separate data that is representable in graph form, from data that is not. 
Then a sketch for a family of  methods of anonymization could be the following:

\begin{enumerate}
\item Transform the graph data into a $k$-anonymous graph;
\item Observe that the $k$-anonymous graph provides a partition of the vertex set (see Section~\ref{sec:4}), hence defining a clustering of the records. 
Now apply a method of anonymization providing $n$-confusion with $n=k$ to the every cluster independently. 
\end{enumerate}
It is not hard to see (remembering that $k$-anonymity is a special case of $n$-confusion), that the result is a method of anonymization that provides $n$-confusion with $n=k$. 
Suitable methods of anonymization could be the following:
\begin{itemize}
\item 
Data distortion by probability distribution~\cite{LiCh}. 
This is a special case of synthetic data generation, 
in which the protected data is generated from what is determined to be the probability distribution of the original data. 
It is very important that the data generator provides non-reversible anonymity (as it is supposed to do). 
Otherwise, the result will not satisfy $n$-confusion. 
In particular, one should be careful with  methods that generate data with the same statistics as the original data, since too many restrictions may lead to a determined system of equations, and consequently, reversible anonymity.
The exact characterization of data generators that provides non-reversible anonymity is still to determine. 
Observe that although the new data asymptotically preserves statistics locally, if no further actions are taken, all global statistics will not be preserved. 
Also observe that if the clusters are small, then local statistics may not be preserved. 

The idea to combine a first step of clustering with a second step of data generation within the clusters, 
has previously been studied for non-graph data in \cite{Josep} and \cite{PAIS2012}, Section 4.

\item Data swapping. In this case, we suggest that all data should be uniformely scrambled, attribute by attribute. Then it is clear that the result satisfies $n$-confusion. Also, some of both the global and the local statistics are preserved, since the data entries are the same as the original data entries. Future work includes an evaluation of the information loss. Possible relaxations could be discussed in order to lower the information loss.   
\end{itemize}
\section{Previous work}
\label{sec:3}
Previous applications of $k$-anonymity to graphs have suggested several different quasi-identifiers, resulting in different definitions of $k$-anonymity for graphs. 
These definitions compete, and no agreement has been reached. We review some of the available definitions here below. 

\subsection{k-Anonymous graphs in terms of structural queries}
Hay et al.~\cite{ref:Hay.Miklau.Jensen.2008} explore the potential of structural queries on graphs for the reidentification of vertices and propose a formalization of the graph anonymization problem based on $k$-anonymity. 
Given a graph $G=(V,E)$ and an anonymization of it, $G'=(V',E')$, they let an adversary post queries on the structure of $G$ in a neigborhood of a fixed vertex $v\in V'$. 
The vertex sets $V$ and $V'$ are assumed to be equal, so that if $v$ is a vertex in the anonymized graph, then it is also a vertex in the original graph. 
They define the candidate set of the vertex $v\in V$ with respect to the query $Q$ as the set of vertices $cand_Q(x)\subseteq V$ such that the outcome of the query is the same for all vertices in $cand_Q(x)$, 
$$cand_Q(x)=\{v\in V:Q(x)=Q(v)\}.$$
As observed by Hay et al. the candidate sets with respect to a fixed query form a partition of the vertice set $V$ into equivalence classes. 
In their model of the behavior of an adversary, he posts a sequence of structural queries. 
The intersection of the results from this sequence of queries is compared with the additional information the adversary has on the vertex he wants to reidentify and may then provide a refinement of the reidentification of a vertex compared to what a single query provides.
A graph is then $k$-candidate anonymous if it satisfies the following condition. 

\begin{definition}~\cite{ref:Hay.Miklau.Jensen.2008}
\label{def:k.candidate.anonymity}
Let $Q$ be a structural query. An anonymized graph satisfies $k$-candidate anonymity given Q if:
$$\forall x \in V, \forall y \in cand_Q(x) : C_{Q,x}[y] \leq 1/k$$
where $C_{Q,x}[y]$ is the probability, given $Q$, of taking candidate $y$ for $x$. 
The authors define $C_{Q,x}[y]=1/{|cand_Q(x)|}$ for each $y \in cand_Q(x)$ and 0 otherwise.
\end{definition}


The anonymization method proposed in \cite{ref:Hay.Miklau.Jensen.2008} is based on the idea of $k$-anonymity as a partition of the record set. 
An algorithm is described that uses simulated annealing in order to find a partition of the vertex set that satisfies the $k$-anonymity constraint and maximizes the descriptive properties of the relations between the classes of the partition. 
The algorithm returns a generalization of the graph; a set of super-vertices, corresponding to the classes of the partition of the vertices,  connected by a set of super-edges (including super self-loops), corresponding to the structural relations between the classes. 
If the partition of the vertices is defined so that $k$-anonymity is obtained with respect to the information contained in the super-edges, then the privacy constraint is indeed satisfied. 
The data owner can choose to publish either the generalized graph or a sampled graph with the same properties as the ones described by the generalized graph. 

The approach in \cite{ref:Hay.Miklau.Jensen.2008} does not fix a unique quasi-identifier, but leaves it up to the data owner to choose which structural attributes are important to publish and/or protect. 
Among the previous work we have found, it is also probably the approach that is closest to the one presented in this article. 

\subsection{k-anonymous graphs with respect to the degree}
If we choose the degree of the vertices as the quasi-identifier of the graph, 
then we obtain the definition of $k$-anonymous graph proposed by Liu and Terzi~\cite{ref:Liu.Terzi.2008}. 
The reason for their proposal seems to be of pragmatic nature. 
Without doubt they are aware of the fact that the degree of the graph is not the only attribute that can be used as a quasi-identifier. However they explore the possibility to anonymize the graph with respect to this sole attribute while making as little changes in the graph as possible, using a greedy method.

Their definition of $k$-anonymity for graphs is as follows.  

\begin{definition}~\cite{ref:Liu.Terzi.2008}
A graph $(V,E)$ is $k$-degree anonymous if every number that appears as a degree of a vertex in $V$, appears as the degree for at least $k$ vertices in $V$.  
\end{definition}

\subsection{k-Anonymous graphs with respect to isomorphic 1-neighborhoods}
If we instead consider that the quasi-identifier of the graph is the induced subgraph of the neighbors of a vertex, then we obtain the definition of $k$-anonymous graph proposed by Zhou and Pei~\cite{ref:Zhou.Pei.2008}. 
Given a graph $G=(V,E)$ and a vertex $v\in V$, the $d$-neigborhood $Neighbor_G^d(v)$ is the induced subgraph of the set of vertices of distance $d$ from $v$. For $d=1$, the $1$-neighborhood $Neighbor_G^1(v)$ is the induced subgraph of the set of vertices that share an edge with $v$. 

\begin{definition}~\cite{ref:Zhou.Pei.2008}
Let $G=(V,E)$ be a graph. Then, the $1$-neighborhood of $v \in V$ is the induced subgraph of the neighbors of $u$, denoted by $Neighbor_G^1(u) = G(N(u))$ where $N(u)=\{v|(u,v)\in E\}$, and where $G(N(u))$ is defined with the vertices $N(u)$ and the edges $E_{N(u)} = \{(u,v) | (u,v) \in E$ and $u \in N(u)$ and $v \in N(u)\}$
\end{definition}

A graph isomorphism between two graphs is a mapping that transforms one graph into the other by reindexing the vertices. 
When there exists a graph isomorphism between two graphs, then we say that the graphs are graph isomorphic. 
The definition of $k$-anonymity for graphs based on isomorphic 1-neighborhoods is as follows. 
\begin{definition}~\cite{ref:Zhou.Pei.2008}\label{def:kanonymzhoupei}
Let $G=(V,E)$ be a graph. The graph $G'=(V',E')$ is a $k$-anonymization of $G$ if $V\subseteq V'$, $E\subseteq E'$ and  for every vertex in $V$ there are at least $(k-1)$ other vertices $v_1,\dots , v_{k-1}\in V$ such that 
$Neighbor_{G'}^1(A(u))$, $Neighbor_{G'}^1(A(v_1))$, $\dots$, $Neighbor_{G'}^1(A(v_{k-1}))$ are isomorphic.
\end{definition}
In \cite{ref:Zhou.Pei.2008} there is also an algorithm to accomplish $k$-anonymity according to Definition~\ref{def:kanonymzhoupei}.

\subsection{(k,l)-Anonymous graphs with respect to subsets of neighborhoods}
Yet another version of $k$-anonymity for graphs has been proposed. 
In this definition, a part from the parameter $k$, an additional parameter $l$ is used. 
The parameter $k$ plays the same role as in $k$-anonymity. 
The definition, proposed by Feder et al.~\cite{ref:Feder.Nabar.Terzi.2008}, is given below. 

\begin{definition}~\cite{ref:Feder.Nabar.Terzi.2008}
A graph $G = (V,E)$ is $(k,l)$-anonymous if for each vertex $v \in V$, there exists a set of vertices $U \subseteq V$ not containing $v$ such that $|U| \geq k$ and for each $u \in U$ the two vertices $u$ and $v$ share at least $l$ neighbors. 
\end{definition}

\subsection{A criticism of (k,l)-anonymity}
With the following example we show that for any pair $(k,l)$ with $k \leq l$ it is possible to find a graph that is $(k,l)$-anonymous, 
but in which reidentification is possible for a large proportion of the vertices, using only two of their neighbor vertices.  

\begin{example}
\label{ex:klanonimous.graph}
Let $k$, $l$, and $m$ be three arbitrary integer numbers such that $k \geq l\geq 1$, and $m>2$. 
Define a graph $G$ with vertices $V=\{v_0, \cdots, v_{k-1}, u_{0}, \cdots, u_{m-1}\}$ and let the edges $E$  be defined as the union of the following sets of edges: 
\begin{description}
\item [(i)] $(v_i, v_j)$ for all $v_i, v_j \in \{v_0, \cdots, v_{k-1}\}$
\item [(ii)]$(u_{i},u_{(i+1) \pmod{m}})$ and $(u_{(i+1) \pmod{m}},u_{i})$ for $i$ in $\{0, \dots, m-1\}$
\item [(iii)] for all $u_i \in \{u_0, \dots, u_{m-1}\}$, include $(u_i, v)$ and $(v, u_i)$ for all $v \in W_i$ where $W_i$ is a subset of $\{v_0, \cdots, v_{k-1}\}$ of cardinality $l$. 
\end{description}
Then, this graph satisfies $(k,l)$-anonymity. 
In addition, it is easy to see that any vertex $v_i$ can be reidentified by the pair of vertices $u_{i-1 \pmod{m}}, u_{i+1 \pmod{m}}$. 
\end{example}

The $k$ vertices $\{v_0,\dots,v_{k-1}\}$ in Example~\ref{ex:klanonimous.graph}, which we denote by $V_1$, form a clique, that is, a subgraph which is complete. 
The $m$ vertices $\{u_0,\dots,u_{m-1}\}$, which we denote by $V_2$, are only connected with two other vertices in $V_2$ and $l$ vertices in the clique $V_1$. 
The graph has $m+k$ vertices and $2 k (k-1)  +  2 m    +    2 m l $ edges. 

Observe that the parameter $m$ can be as large as desired. 
Therefore, we can make the proportion of vertices that can be reidentified by only two neighbors close to 1. 
Also note that the selection of $k$ and $l$ is arbitrary. 
The only requirement is that $k \geq l$. 
In case of $k < l$, it is easy to see that an $(l,l)$-anonymous graph constructed as in the example above will satisfy $(k,l)$-anonymity. 

\begin{example}
Figure~\ref{fig:klanonimity.k=8.l=3} illustrates the construction of a $(k=8,l=3)$-anonymous graph with $m=12$ vertices in the border following the construction of Example~\ref{ex:klanonimous.graph}. 
\end{example}

\begin{figure}[t]
\begin{center}
\includegraphics[angle=0, width=5cm]{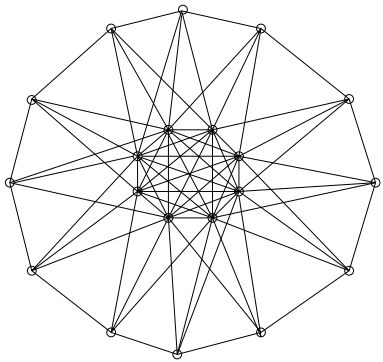}
\end{center}
\caption[]{Example of a graph satisfying $(k=8,l=3)$-anonymity but where information about two neighbours of any of the $m=12$ nodes in the border of the graph reidentifies it}
\label{fig:klanonimity.k=8.l=3}
\end{figure}

\section{A definition of k-anonymity for graphs}
\label{sec:4}
In this section we will present the definition of $k$-anonymity which we consider to be the appropiate for graphs. 

Let $(V,E)$ be a graph and let $v$ be a vertex in $V$. 
Define the neighbors of $v$ as the set of vertices of distance one to $v$, that is, the set of vertices 
$$N(v):=\{u\in V:(u,v)\in E\}.$$
We give the following definition of $k$-anonymous graph. 
\begin{definition}\label{def:kanonymVK}
Let $G=(V,E)$ be a graph. We say that $G$ is $k$-anonymous if for any vertex $v_1$ in $V$, there are at least $k$ distinct vertices $\{v_i\}_{i=1}^{k}$ in $V$, 
such that $N(v_i)=N(v_1)$ for all $i\in [1,k]$. 
\end{definition}
This definition of $k$-anonymity is appropiate for a data owner that has determined the quasi-identifier of the graph to be the sets of neighbors of the vertices. 
A graph that is $k$-anonymous following this definition has an adjacency matrix in which every row vector appears at least $k$ times.  
The adjacency matrix is a lossless representation of the graph, in the sense that it completely determines the graph. 
Therefore it can be deduced that a graph that is $k$-anonymous following Definition~\ref{def:kanonymVK} is $k$-anonymous with respect to any other graph property, since these properties are implicitly present in the matrix. 
In such a graph, there is a partition of the set of vertices, such that the vertices in the same part share exactly the same neighbors. 
Observe that they do not only share neighbor set, but that they also share non-neighbor set. 

\subsection{Characterizing the k-anonymous graphs} 
In this section we study the characterization of $k$-anonymous graphs. We present two propositions that establish conditions for this types of graphs. First, we consider regular $k$-anonymous graphs and then non-regular ones. 

\subsubsection{Regular k-anonymous graphs}

Let $G=(V,E)$ be a regular graph of degree $d$ that is $k$-anonymous. 
The subsets of vertices that share their neighbors form a partition of the vertex set of the graph, 
$$V=V_1\cup\cdots\cup V_n$$
such that for $i\neq j$ we have $N(v)=N(u)$ for all $v,u\in V_i$ and $N(v)\neq N(w)$ for all $v\in V_i$ and $w\in V_j$. 
However, note that in general it is not true that $N(v)\cap N(w)=\emptyset $ for $v\in V_i$ and $w\in V_j$ when $i\neq j$. 

Choose a vertex $v_1\in V$ and let $V_1=\{v_1,\dots,v_n\}$ be the other vertices in $V$ such that $N:=N(v_1)=N(v_i)$ for all $i\in [2,n]$. 
Because $G$ is $k$-anonymous, $n$ is larger than $k$. 
Let $\{u_j\}_{j=1}^{d}$ be the $d$ vertices in $N$. 
For any $j\in [1,d]$, we then have that $v_i$ belongs to $N(u_j)$ for all $i\in [1,k]$ and $N(u_j)$ has cardinality $d$ so that $k\leq d$.


Fix $u_{j_0}$ for some $j_0\in [1,d]$. 
Then there are vertices $\{w_s\}_{s=2}^{k}$ in $V$ such that $N(u_{j_0})=N(w_s)$ for all $s\in [2,k]$. 
In particular, $v_i$ belongs to  $N(w_s)$ for all $i\in [1,k]$ and all $s\in [2,k]$, so that $w_s\in N(v_i)$ and therefore $w_s\in \{u_j\}_{j=1}^{d}=N$ for all $s\in [2,k]$. 
This implies that any equivalence class $V_i$ of vertices sharing neighbors that contains one of the vertices $u_j$, is contained in the set of vertices $N=\{u_j\}_{j=1}^d$. 
Therefore there is a partition of $N$ into one or several equivalence classes.
The classes of this partition is a subset of the classes of the partition of the whole vertex set.  

If in addition we assume that all the equivalence classes of the partition of the vertex set have the same cardinality $|V_i|=k$, then we can deduce that $k$ divides $d$. 
In this case, $k$ of course also divides the order $|V|$ of the graph. 

Observe that in general the equivalence classes do not need to have the same cardinality. 
The only requirement is that the cardinalities should all exceed $k$ and that the cardinalities of the equivalence classes in a neighbor set $N$ sum up to $d$. 
One can for example construct a regular graph of degree 7 that is 3-anonymous, 
by splitting the neighbor sets into equivalence classes of cardinality 3 and 4. 

We collect the results in this section in the following Proposition~\ref{prop:properties_regular_kanonymvk}. 
 
\begin{proposition}\label{prop:properties_regular_kanonymvk}
Let $G=(V,E)$ be a $d$-regular, $k$-anonymous graph according to Definition~\ref{def:kanonymVK}. 
Then the following conditions are satisfied.
\begin{itemize}
\item  $k\leq d$;
\item 
There is a partition of $V=V_1\cup\cdots\cup V_n$
such that for $i\neq j$ we have $N(v)=N(u)$ for all $v,u\in V_i$ and $N(v)\neq N(w)$ for all $v\in V_i$ and $w\in V_j$.
However, in general $N(v)\cap N(w)\neq \emptyset$ for $v\in V_i$ and $w\in V_j$;
\item Fix a vertex $v\in V$. Then there is a partition of $N(v)$ into one or several equivalence classes.
The classes of this partition form a subset of the classes of the partition of the whole vertex set; 
\item If all the equivalence classes $V_i$ have the same cardinality $|V_i|=k$, 
then $k$ divides $d$ and $|V|$. 
\end{itemize}
\end{proposition}

\subsubsection{Non-regular k-anonymous graphs}
In non-regular $k$-anonymous graphs, the vertex set of the graph will also form a partition of classes of vertices that share neighbors. By definition, the cardinality of these classes will be at least $k$. 

Many of the arguments we have made for regular $k$-anonymous graphs are true also for non-regular $k$-anonymous graphs. 
Let $N$ be the neighbor set of one class $V_i$ of vertices. 
Then $V_i$ are neighbors to the vertices in $N$. 
In a connected graph all vertices have neighbors, and therefore the minimum degree of a connected $k$-anonymous graph must be larger than $k$. 

We also see that, by the same argument as in the regular case, fixed a vertex $u\in N$, the partition $V_j$ to which it belongs must contain only vertices in $N$. 
Therefore, also in this case the neighbor sets of the graph are partitioned into equivalence classes of vertices that share the same neighbors. 

As commented previously, when $k$ is large compared to the order of the graph $|V|$, then there are only a few graphs of that order that are $k$-anonymous. 
This implies that if we protect a graph of small order so that it gets $k$-anonymous for a large $k$, then the information loss is important. 
However, if $k$ is small and the order of the graph is large, then much information is still kept. 

The following Proposition~\ref{prop:properties_kanonymvk} collects the results presented in this section and is a generalization of Proposition~\ref{prop:properties_regular_kanonymvk}.
 
\begin{proposition}\label{prop:properties_kanonymvk}
Let $G=(V,E)$ be a $k$-anonymous graph according to Definition~\ref{def:kanonymVK}. 
Then the following conditions are satisfied.
\begin{itemize}
\item The minimum degree of $G$ is larger than $k$;
\item 
There is a partition of $V=V_1\cup\cdots\cup V_n$
such that for $i\neq j$ we have $N(v)=N(u)$ for all $v,u\in V_i$ and $N(v)\neq N(w)$ for all $v\in V_i$ and $w\in V_j$.
However, in general $N(v)\cap N(w)\neq \emptyset$ for $v\in V_i$ and $w\in V_j$;
\item Fix a vertex $v\in V$. Then there is a partition of $N(v)$ into one or several equivalence classes.
The classes of this partition form a subset of the classes of the partition of the whole vertex set.  
\end{itemize}
\end{proposition}

\section{A definition of (k,l)-anonymity for graphs as a relaxation of $k$-anonymity for graphs}
\label{sec:5}
The definition we just presented has the problem that it is sometimes rather restrictive, in particular for small data sets. 
Observe that if $k$ is large in relation to the order $|V|$ of the graph, 
then there is only a small number of non-isomorphic graphs that will satisfy the criterion of $k$-anonymity. 
Under these circumstances, the usefulness of the anonymized graph is therefore limited. 
This fact justifies the following relaxation of the definition of $k$-anonymity for graphs. 
Following the idea in \cite{ref:Feder.Nabar.Terzi.2008}, 
we introduce a second parameter $l$, 
and consider that the graph is $(k,l)$-anonymous if it is $k$-anonymous with respect to any subset of cardinality at most $l$ of the neighbor sets of the vertices of the graph.
The phrase ``a subset of cardinality at most $l$ of the neighbor sets of vertices'' has two distinct interpretations, 
resulting in two distinct definitions of $(k,l)$-anonymity for graphs. Which of the two definitions should be used, depends on the context. 
 
If the graph has no multiple edges (a pair of vertices can be connected by at most one edge), then the row vector in the adjacency matrix that represents the neighbor set $N(v)$ of a vertex $v$ is a vector in the space $\{0,1\}^n$, where $|V|=n$ is the number of vertices in the graph. 
If the graph has multiple edges, then the space is $(\mathbb{N}\cup\{0\})^n$. 
Interpret a subset of cardinality at most $l$ of the neighbor set of $v$ to be the entries of the vector $N(v)$ which are indexed  by a subset of the indices of $N(v)$ of cardinality at most $l$.
In this way we characterize the vertex $v$ by subsets of both its neighbors and its non-neighbors and we give the following definition of $(k,l)$-anonymity for graphs.

\begin{definition}[$(k,l)$-anonymity for graphs (I)]\label{def:klanonymVK1}
Let $G=(V,E)$ be a graph. We say that $G$ is $(k,l)$-anonymous if for any vertex $v_1\in V$ and 
for all subset of indices $I\subseteq [1,|N(v_1)|]$ of cardinality $|I|\leq l$ there are at least $k$ distinct vertices $\{v_i\}_{i=1}^{k}$ such that $N(v_i)$ and $N(v_1)$ coincide over $I$ for $i\in[1,k]$. 
\end{definition}

In a graph that satisfies Definition~\ref{def:klanonymVK1}, 
an adversary who fixes a vertex $v$ for reidentification and who has access to the induced subgraph on a subset of vertices of the graph of at most cardinality $l$ as auxiliary information, 
will only be able to reidentify $v$ with probability at most $1/k$. 

In contrast to Definition~\ref{def:klanonymVK1}, we could interpret a subset of cardinality at most $l$ of the neighbor set of $v$ to be a subset of the neighbors of $v$. 
In this way the non-neighbors are not taken into account.  
Using this interpretation we give the following formal definition of $(k,l)$-anonymity. 
 
\begin{definition}[$(k,l)$-anonymity for graphs (II)]\label{def:klanonymVK2}
Let $G=(V,E)$ be a graph. 
We say that $G$ is $(k,l)$-anonymous if for any vertex $v_1\in V$ and 
for all subset $S\subseteq N(v_1)$ of cardinality $|S|\leq l$ there are at least $k$ distinct vertices $\{v_i\}_{i=1}^{k}$ such that $S\subseteq N(v_i)$ for $i\in[1,k]$. 
\end{definition}

In a graph that satisfies Definition~\ref{def:klanonymVK2}, an adversary who fixes a vertex $v$ for reidentification and who has access to the induced subgraph of a subset of vertices of the graph that contains at most $l$ of the neighbors of $v$ as auxiliary information, will only be able to reidentify $v$ with probability at most $1/k$. 

Observe that the fact that Definition~\ref{def:klanonymVK1} and Definition~\ref{def:klanonymVK2} are relaxations of Definition~\ref{def:kanonymVK}, implies that a graph that satisfies $(k,l)$-anonymity is not in general $k$-anonymous, and that one could even find examples of $(k,l)$-anonymous graphs in which every vertex is uniquely identified by some property, say, by their degree. 

This observation means that in a situation where the data owner considers that there is an elevated risk that the adversary could have access to some auxiliary information besides a subgraph containing at most $l$ neighbors of any vertex, further protection is recommended. 
For example, in the case when the additional auxiliary information consists of the degrees of the vertices of the graph, the data owner could consider a graph protection method which combines $(k,l)$-anonymization and $k$-anonymization with respect to the degree. The latter method can be found in \cite{ref:Liu.Terzi.2008}. 

However, whenever the auxiliary information about the graph that is available to the adversary is restricted to
\begin{itemize}
\item the induced subgraph of the original graph on at most $l$ vertices in the case of Definition~\ref{def:klanonymVK1}, or
\item the induced subgraph of the original graph on a subset of the original vertices that contains at most $l$ neighbors of any of the original vertices, in the case of Definition~\ref{def:klanonymVK2},
\end{itemize}
then the information that he has about the degrees of the vertices is equally restricted, 
so that the data protection in a $(k,l)$-anonymous graph is just as high as it claims to be. 

It is obvious from the definition that $k$ must be smaller than the order $|V|$ of the graph. 
If we assume that the graph contains no loops, 
then the set of $k$ vertices that share neighbors and the set of $l$ neighbors that they share must be disjoint,
so that $l$ can not be larger than $|V|-k$. 
These bounds are attained, since the complete graph $(V,E)$ on $n:=|V|$ vertices, is $(k,n-k)$-anonymous for all $k\leq n$.   

Consider a graph that is $(k,l)$-anonymous according to Definition~\ref{def:klanonymVK2}. 
Let $d$ be the minimum degree of the graph. 
If $d$ is smaller than $l$, 
so that there is a vertex $v$ with a smaller number of neighbors than $l$, 
then $v$ can share at most $d$ neighbors with other vertices, 
so that the graph can be at most $(k,d)$-anonymous. 
Therefore, for $(k,l)$-anonymity according to Definition~\ref{def:klanonymVK2} it should always be assumed that $l\leq d$. 

Observe that if a vertex share a set of $l$ neighbors with $k$ other vertices, 
then it also shares any subset of these $l$ neighbors with the same $k$ vertices. 
Therefore, if a graph is $(k,l)$-anonymous according to Definition~\ref{def:klanonymVK2}, 
then it is also $(k,l')$-anonymous according to Definition~\ref{def:klanonymVK2} for all $l'\leq l$.  

Also, if a vertex share a set of $l$ neighbors with $k$ other vertices, 
then it also shares the same set of $l$ neighbors with any subset of the $k$ vertices. 
Therefore, if a graph is $(k,l)$-anonymous according to Definition~\ref{def:klanonymVK2}, then it is also $(k',l)$-anonymous according to Definition~\ref{def:klanonymVK2} for all $k'\leq k$. 

We collect these results in the following Proposition~\ref{prop:properties_klanonymvk2}. 
 
\begin{proposition}\label{prop:properties_klanonymvk2}
Let $G=(V,E)$ be a $(k,l)$-anonymous graph, following Definition~\ref{def:klanonymVK2}. 
Then the following conditions are satisfied.
\begin{itemize}
\item If $G$ has no loops, then $k+l\leq |V|$;
\item $l$ is smaller or equal to the minimum degree of the graph;
\item $G$ is $(k,l')$-anonymous for all $l'\leq l$;
\item $G$ is $(k',l)$-anonymous for all $k'\leq k$. 
\end{itemize}
\end{proposition}

\section{Algorithms for the k-anonymization of graphs}
\label{sec:6}
In this section we will present three different algorithms.  
The first is an algorithm for $k$-anonymization of databases. 
The second algorithm determines the degree of $(k,l)$-anonymity of a given graph, that is, given a $k$ it determines the largest $l$ for which the graph is $(k,l)$-anonymous.
The third algorithm increases the degree of $(k,l)$-anonymity of a graph. More precisely, if the algorithm is given a graph that is $(k,l)$-anonymous, then it returns a similar graph that is $(k,l')$-anonymous, with $l'>l$. 
\subsection{A k-anonymization algorithm}
As commented in Section~\ref{sec:4}, it is easy to see that in a graph that is $k$-anonymous according to Definition~\ref{def:kanonymVK} there exists a partition of the vertex set in classes of cardinality at least $k$, 
so that the vertices in a class of the partition all share the same neighbors. 
This partition is easy to find. 
The row vectors (or, equivalently, the column vectors) of the adjacency matrix of the graph represent the neighbor set of their corresponding vertex, so that the vertices in the same class of a $k$-anonymous graph will have equal row vectors (and column vectors). 
Since every class contains at least $k$ vertices, 
the adjacency matrix of a $k$-anonymous graph is a table that satisfies $k$-anonymity. 

Now suppose that we have a graph with an adjacency matrix $A$ that does not satisfy $k$-anonymity and that we want to transform $A$ in order to obtain another table $A'$ that is similar to $A$, satisfies $k$-anonymity and is the adjacency matrix of a graph. 
That is, suppose that we want to define a method for $k$-anonymization of graphs. 
In this section we give an algorithm (Algorithm~\ref{alg:1}) that describes such a method. 

The algorithm is based on a clustering algorithm for graphs. 
We require that the clustering algorithm returns a partition of the vertex set $V$ of the graph and that each cluster or class of vertices contains at least $k$ vertices. 
In \cite{ref:Hay.Miklau.Jensen.2008b}, the authors use simulated annealing in order to find a partition of the vertices that satisfy $k$-anonymity and minimizes information loss, via a maximum likelihood
approach.
Heuristic methods are nice, because they work. However, other methods may offer more theoretical control over the properties of the output of the algorithm. 

In order to obtain good clustering results the choice of 
the distance to use is crucial. 
Not only classical distances are used for clustering, but also weaker topology concepts, like similarities~\cite{simref,simref2} and proximity relations~\cite{ref:Klir.Yuan.1995}. 
The paper \cite{IEEE_Taipei} describes some available algorithms for clustering of graphs, 
and discusses how to define similarities between vertices. 
In particular, two similarities for clustering of the
neighbor sets of the vertices of a graph are compared; the Manhattan similarity, based on the Manhattan distance, 
and the so-called 2-path similarity. 
The Manhattan similarity measures how many
equal entries the two vertices have in the adjacency matrix.
Formally, we have the following definition. 
\begin{definition}
Given two vectors in the adjacency matrix of $G$, $u,v \in \{0,1\}^{|V|}$, we denote by $sim_{l_1}(u,v)$ the Manhattan or $l_1$ similarity between $u$ and $v$, so that $$sim_{l_1}(u,v)=|V|-\sum_{i=1}^{n}|u[i]-v[i]|.$$
\end{definition}
The 2-path similarity measures the number of paths of length 2 between the two vertices. 
This can be calculated by taking the square of the adjacency matrix, and we define the 2-path similarity in this manner. 
\begin{definition}
Given two vectors in the adjacency matrix of $G$, $u,v \in \{0,1\}^{|V|}$, we denote by $sim_{2-path}(u,v)$ the 2-path similarity between $u$ and $v$, so that $$sim_{2-path}(u,v)=\sum_{i=1}^{n}u[i]v[i].$$
\end{definition}
It is interesting to note how two vertices that share many neighbors have many paths of length two between them, or, expressed in another way, they have many quadrangles that passes through them.
As explained in \cite{IEEE_Taipei}, the Manhattan similarity measures the similarity between
vertices with respect to both neighbors and non-neighbors, while the 2-path
similarity only measures the similarity between vertices with respect to their
neighbors, so that a common non-neighbor does not change the similarity between
two vertices. 
The differences between these two similarities should be compared with the differences between the two
definitions of $(k, l)$-anonymity in Definition~\ref{def:klanonymVK1} and Definition~\ref{def:klanonymVK2}. 
Algorithm~\ref{alg:2} and Algorithm~\ref{alg:3} illustrates the connection between the similarities and the definitions of $(k,l)$-anonymity. 

A clustered graph can be published as a generalized graph as described in \cite{ref:Hay.Miklau.Jensen.2008}. 
When the graph is $k$-anonymous, it is possible to publish a generalization of the graph that satisfies $k$-anonymity, without any information loss at all. 
This is indeed the idea behind $k$-anonymity for graphs. 

However, when the graph is not $k$-anonymous, a generalization of the graph that satisfies $k$-anonymity is never lossless. 
We want a method to transform any graph into a protected graph that satisfies $k$-anonymity according to Definition~\ref{def:kanonymVK}. 
We also want the protected graph to be similar to the original graph, so that the information loss is kept small.

\begin{algorithm}[t]
\KwIn{A graph $G=(V,E)$ and a natural number $k\leq|V|$}
\KwOut{A graph $G'=(V, E')$ that is $k$-anonymous according to Definition~\ref{def:kanonymVK}}
Calculate the matrix $S=(s_{ij})_{i,j=1}^{|V|}$ with 
$s_{ij}:=sim(v_i,v_j)$ for $v_i,v_j\in V$\;
Partition the rows of $S$ in clusters, hence obtaining a family of clusters $C$ of $V$\;
\ForEach{$C_i,C_j\in C$}{
\If{$\sharp\{(u,v)\in C_i\times C_j: \exists uv\in E\}> |C_i||C_j|/2$}
{
\ForEach{$(u,v)\in C_i\times C_j$}
{
\If{$uv\not\in E$}{
Add $uv$ to $E$\;
}
}
} 
\Else
{
\ForEach{$(u,v)\in C_i\times C_j$}
{
\If{$uv\in E$}
{
Delete $uv$ from $E$\;
}
}
}
}
Return $G$\;
\caption{An algorithm for $k$-anonymization of graphs using clustering and plurality rule}
\label{alg:1}
\end{algorithm}

\subsection{Algorithms for the calculation of the degree of (k,l)-anonymity of a graph, given k}
If we calculate the Manhattan similarity between all the vertices in the graph, then we can determine the highest $l$ such that the graph is $(k,l)$-anonymous according to Definition~\ref{def:klanonymVK1}, for $k$ fixed. 
This $l$ define neighborhoods or balls around each vertex $v$; set of vertices $\{u\}$ that satisfy $sim_{l_1}(v,u)\geq l$, with at least $k$ vertices in each neighborhood. 
When $k=1$, then the largest $l$ that gives us this family of neighborhoods is trivially equal to the order of the graph $l=|V|$, 
When $k>1$, then $l$ might be smaller, somewhere between 0 and the order of the graph. 
The cardinality of the set of neighborhoods is $|V|$ and it forms a fuzzy clustering of the neighbor sets of the vertices of the graph. 
The centroids of these clusters are the points $\left(\mathbb{Z}/(2)\right)^{|V|}$ that represent the neighbor sets, and since every cluster has cardinality $k$, it is obvious that whenever $k>1$ then the clusters overlap. 

If we instead use the 2-path similarity, then we can determine the largest $l$ so that the graph is $(k,l)$-anonymous according to Definition~\ref{def:klanonymVK2}. 
For $k=1$, the largest $l$ is equal to the minimum degree of the graph, and if $k>1$, then the largest possible $l$ is somewhere between 0 and the minimum degree. 

In this way we obtain two different measures on the degree of anonymity of the original graph; the largest parameters so that the graph satisfies the definitions of the two versions of $(k,l)$-anonymity. 
Which of the two measures is the most useful, depends on the context, or more precisely, it depends on if non-neighbors are useful for reidentification or not. 

Hence, we present here an algorithm (Algorithm~\ref{alg:2}) that, given a graph $G=(V,E)$ and a positive integer $k$, calculates the largest $l$ such that $G$ is $(k,l)$-anonymous according to Definition~\ref{def:klanonymVK1}. 
The algorithm shows the relation between the Manhattan or $l_1$ similarity \cite{IEEE_Taipei} and the $(k,l)$-anonymity according to Definition~\ref{def:klanonymVK1}, described before. 
 
\begin{algorithm}[t]
\KwIn{A graph $G=(V,E)$ and a natural number $k<|V|$}
\KwOut{The largest $l$ such that $G$ is $(k,l)$-anonymous according to Definition~\ref{def:klanonymVK1}}
$s:=|V|$\;
\While{Exists $v\in V$ such that $\sharp\{u:sim_{l_1}(u,v)\geq s\}<k$}
{$s--$;}
Return $l:=s$\;
\caption{An algorithm that given $G$ and $k$ computes the largest $l$ for which $G$ is $(k,l)$-anonymous, using the Manhattan similarity}
\label{alg:2} 
\end{algorithm}

Next we present an algorithm (Algorithm~\ref{alg:3}) that, given a graph $G=(V,E)$ and a positive integer $k$, calculates the largest $l$ such that $G$ is $(k,l)$-anonymous according to Definition~\ref{def:klanonymVK2}. 
The algorithm shows the relation between the 2-path similarity \cite{IEEE_Taipei} and the $(k,l)$-anonymity according to Definition~\ref{def:klanonymVK2}, as described before. 

\begin{algorithm}[t]
\KwIn{A graph $G=(V,E)$ and a natural number $k<|V|$}
\KwOut{The largest $l$ such that $G$ is $(k,l)$-anonymous according to Definition~\ref{def:klanonymVK2}}

\ForEach{$v_i\in V$}{$s[i]:=\mbox{Degree}(v_i)$\;}
\While{Exists $v_i\in V$ such that $\sharp\{u\in V:sim_{2-path}(u,v_i)\geq s[i]\}<k$} 
{$s[i]--$\; 
}
Return $l:=\min_i(s[i])$ \; 
\caption{An algorithm that given $G$ and $k$ computes the largest $l$ for which $G$ is $(k,l)$-anonymous, using the 2-path similarity}
\label{alg:3}
\end{algorithm}

\subsection{An algorithm to increase the degree of (k,l)-anonymity of a graph}
Finally we present an algorithm (Algorithm~\ref{alg:4}) that, given a graph $G$ that is $(k,l)$-anonymous with respect to the similarity $sim$, returns either a graph $G'$ that is based on $G$ but that is $(k,l')$-anonymous with $l'=l+1$ or the empty graph without vertices.  
\begin{algorithm}[t]
\KwIn{A graph $G=(V,E)$ and a natural number $k<|V|$}
\KwOut{A graph $G'=(V',E')$ such that if $G$ is $(k,l)$-anonymous, then  $G'$ is $(k,l')$-anonymous  with $l'=l+1$}

Calculate the largest $l$ such that $G$ is $(k,l)$-anonymous, using Algorithm~\ref{alg:2} or Algorithm~\ref{alg:3}\;
$l':=l+1$\;
\While{Exists $v\in V$ such that $\sharp\{u\in V:sim(u,v)\geq l'\}<k$}
{
\If{$uv\in E$}{Delete $uv$ from $E$\;}
Delete $v$ from $V$\;
}
Return $G$\;
\caption{An algorithm that given a $(k,l)$-anonymous graph $G$ returns a $(k,l+1)$-anonymous graph}
\label{alg:4}
\end{algorithm}

The algorithm deletes vertices $v$ that have a set of similar vertices $\sharp\{u\in V:sim(u,v)\geq l'\}$ which is too small, that is, smaller than $k$. 
Since the deletion of a vertex affects the neigbor sets of the other vertices $v$ so that $\sharp\{u\in V:sim(u,v)\geq  l'\}$ may decrease, causing the deletion of more vertices in the next execution of the loop, there is a risk that the algorithm deletes all vertices of the graph. 
In order to avoid this phenomenon, we recommend, for some $v\in V$ with $\sharp\{u\in V:sim(u,v)\geq l'\}<k$, the addition of new vertices to  $V$ to augment $\sharp\{u\in V:sim(u,v)\geq l'\}$. 
Such a new vertex $\tilde{v}$ must be connected to the already existing vertices in $V$ in a way so that the set $\sharp\{u\in V:sim(u,\tilde{v})\geq l'\}\geq k$, so that we do not introduce new problematic vertices. 
An easy  solution to this problem is to let the new vertices $\tilde{v}$ be copies of the problematic vertex $v$. 
However, the use of this solution would cause the algorithm to loose its status as anonymization method, since the records in the resulting table would not correspond to distinct individuals, so that it fails to be a database according to our definition. 
As a consequence, the risk of reidentification for the vertices protected in this way will be higher than for the vertices in a graph that is protected using a real method of $(k,l')$-anonymization method. 

\section{Conclusion}
\label{sec:7}
In this article we have provided a formal framework for reidentification in general. 
We have defined $n$-confusion as a concept for modelling the anonymity of a database table and we have proved that $n$-confusion is a generalization of $k$-anonymity. 
Then after a short survey on the different available definitions of $k$-anonymity for graphs we provided a new definition for $k$-anonymity, which we consider to be the correct definition. 
It has been explained how this definition can be used in combination with $n$-confusion, for the anonymization of data from, for example, social networks. 

We have provided a description of the $k$-anonymous graphs, both for the regular and the non-regular case. 
However, under some conditions our definition of $k$-anonymity is quite strict, 
so that it is only satisfied by a small number of graphs. 
In order to avoid this problem, 
we have introduced the more flexible definition of $(k,l)$-anonymity. 
Our definition of $(k,l)$-anonymity for graph is meant to replace the definition in \cite{ref:Feder.Nabar.Terzi.2008}, which we have proved to have severe weaknesses. 
We have given two variants of the definition of $(k,l)$-anonymity, which may serve under different conditions. 

We have also provided a set of algorithms;
one algorithm that given a graph $G$ and a natural number $k$ returns a graph $G'$ based on $G$ that is $k$-anonymous, 
two algorithms that given a graph $G$ and a natural number $k$ calculates the largest $l$ such that $G$ is $(k,l)$-anonymous according to our two different definitions of $(k,l)$-anonymity, and finally, 
one algorithm that given a graph $G$ that satisfies $(k,l)$-anonymity returns a graph $G'$ similar to $G$ that satisfies $(k,l+1)$-anonymity. 

\section*{Acknowledgements}

Partial support by the Spanish MEC projects ARES (CONSOLIDER INGENIO 2010 CSD2007-00004), eAEGIS (TSI2007-65406-C03-02 ),  COPRIVACY (TIN2011-27076-C03-03), and RIPUP (TIN2009-11689) is acknowledged. One author is partially supported by the FPU grant (BOEs 17/11/2009 and 11/10/2010) and by the Government of Catalonia under grant 2009 SGR 1135. The authors are with the UNESCO Chair in Data Privacy, but their views do not necessarily reflect those of UNESCO nor commit that organization.

\end{document}